\documentclass[11pt,a4paper]{article}

\usepackage[T1]{fontenc}
\usepackage{newtxtext}
\usepackage[margin=1in]{geometry}
\usepackage{graphicx,booktabs,array,tabularx}
\usepackage{amsmath,amsthm,mathtools}
\usepackage{newtxmath}
\usepackage{enumitem,float,caption,microtype,xcolor}
\usepackage[round,authoryear]{natbib}
\usepackage[colorlinks=true,linkcolor=blue!45!black,
  citecolor=blue!45!black,urlcolor=blue!45!black]{hyperref}
\hypersetup{
  pdftitle={The Risk-Neutral Crash Frontier: Sharp Joint Bounds on Crash Probability and Conditional Depth from Option Bid-Ask Quotes},
  pdfauthor={Jirong Zhuang},
  pdfsubject={Model-free option-implied crash probability and conditional depth},
  pdfkeywords={risk-neutral crash frontier, option-implied tail risk, partial identification, model-free bounds}
}

\graphicspath{{./}}
\newtheorem{theorem}{Theorem}[section]
\newtheorem{proposition}[theorem]{Proposition}
\newtheorem{lemma}[theorem]{Lemma}
\newtheorem{corollary}[theorem]{Corollary}

\newtheorem{assumption}[theorem]{Assumption}
\theoremstyle{remark}
\newtheorem{remark}[theorem]{Remark}

\newcommand{\Q}{\mathbb{Q}}
\newcommand{\R}{\mathbb{R}}
\newcommand{\E}{\mathbb{E}}
\newcommand{\one}{\mathbf{1}}
\newcommand{\cM}{\mathcal{M}}
\newcommand{\cC}{\mathcal{C}}
\newcommand{\cB}{\mathcal{B}}
\newcommand{\cR}{\mathcal{R}}
\newcommand{\cT}{\mathcal{T}}
\newcommand{\cY}{\mathcal{Y}}
\newcommand{\SampleWeeks}{557}
\newcommand{\MainCases}{2,228}
\newcommand{\HardShare}{93.90\%}
\newcommand{\InteriorQualifiedCases}{2,076}
\newcommand{\QualifiedCases}{1,966}
\newcommand{\PrimaryAreaRatio}{0.6340}

\title{\textbf{The Risk-Neutral Crash Frontier:}\\
\large Sharp Joint Bounds on Crash Probability and Conditional Depth from Option Bid--Ask Quotes}
\author{Jirong Zhuang\\
\normalsize Department of General Education, Guangzhou Polytechnic University\\
\normalsize 1342 Shiliang Road, Panyu District, Guangzhou, Guangdong 511483, China\\
\normalsize\texttt{zhuangjirong6552@outlook.com}}
\date{October 2026}

\begin{document}
\maketitle

\begin{abstract}
Index put prices are the market's quotes for crash insurance, and a put's value equals the probability of a crash times the expected shortfall given one.  The market therefore prices the product of likelihood and depth, not the factors, and finitely many bid and ask quotes leave a range of ways to split it.  Fitting one density hides that range, and bounds computed one factor at a time can combine into scenarios that no single risk-neutral distribution could produce.  We characterize the closure of the probability and loss pairs generated by distributions that price every quote inside its spread, with depth as their ratio, and call its boundary the risk-neutral crash frontier.  Under a strict quote interior condition, a partition at the quoted strikes and the threshold represents the closed joint set as the exact projection of a finite linear system, with no price grid.  Linear programs trace the frontier and give sharp bounds for portfolios of the threshold digital and put, with static super- and sub-replicating portfolios of cash, forward, and quoted options certifying the upper and lower bounds.  In weekly SPX cross sections from 2013 to 2023, joint restrictions exclude a median 36.6\% of the area of a benchmark formed from separate probability and unconditional-loss bounds, adjusted for $0\leq L\leq Kp$ and measured in probability--depth coordinates.  Extending the quote set from the eight strikes nearest the threshold to the complete put wing reduces the joint area by a median 5.4 to 18.2\% across the four specifications.
\end{abstract}



\section{Introduction}

Index put prices are the market's quotes for crash insurance.  Dealers use them to price and hedge claims that pay only in a crash, and a large literature reads crash probabilities and tail expectations from them \citep{backus2011disasters,bollerslev2011tails}.  A put price, however, answers two questions at once: how likely is a crash, and how deep will it be if it comes?  Let $X$ be the terminal index level divided by its forward, fix a threshold $K<1$, and under a risk-neutral measure $\Q$ define
\[
p_K=\Q(X\leq K),\qquad L_K=\E_{\Q}(K-X)^+,\qquad s_K=\frac{L_K}{p_K} =K-\E_{\Q}[X\mid X\leq K], \qquad p_K>0.
\]
The loss $L_K$ is the normalized put value.  It averages the shortfall over all states, including those in which the index never falls below $K$, whereas the depth $s_K$ conditions on the crash and asks how far below $K$ the index lands on average.  The identity $L_K=p_Ks_K$ says that the market prices the product of the two and not the factors.

A finite set of strikes quoted at bid and ask pins down neither factor.  A frequent shallow breach and a rare deep one can carry the same price: a put worth $0.10\%$ of forward is consistent with a $2\%$ crash probability and a $5\%$ average shortfall, and equally with a $0.5\%$ probability and a $20\%$ shortfall.  The usual response fits one smooth density through the quotes and reports one crash probability, a number that blends what the market says with what the fitting method adds.  Bounding probability and depth one at a time keeps the quotes but throws the trade-off away.  The distribution that makes a crash least likely need not be the one that makes it deepest.  A stress scenario that pairs the two extremes can therefore be one that the quotes rule out, and a price range for a crash-contingent claim built from separate bounds is wider than the market supports.

We therefore study the two quantities together.  Let $\cM$ be the set of distributions on $\R_+$ with unit mean that price every retained quote inside its interval, and let
\[
\cR_K = \overline{\left\{ \left(\Q(X\leq K),\E_\Q(K-X)^+\right): \Q\in\cM \right\}}
\]
be the closure of the probability and loss pairs these distributions generate.  We call the boundary of $\cR_K$ the risk-neutral crash frontier, or simply the crash frontier.  We use loss rather than depth as the second coordinate because $\cR_K$ is convex while its depth image need not be; since each pair $(p,L)$ with $p>0$ maps to $(p,L/p)$, the depth frontier follows exactly once the boundary of $\cR_K$ is known.

The joint set answers questions that marginal intervals cannot.  A desk that holds a digital option together with a put can ask for the sharp range of the package's value, which is a support function of the set.  Each supporting upper bound is backed by a static portfolio of cash, forward, and quoted options whose cost equals the bound and whose payoff dominates the target.  A risk committee that wants to stress a particular scenario, say a $2\%$ crash probability with a $20\%$ average shortfall, can check whether distributions that price today's quotes can produce or approach it arbitrarily closely; this is equivalent to membership in the joint set.  Comparing the frontier across nested sets of strikes shows whether the deep out-of-the-money put wing carries information about the tail beyond what strikes near $K$ already say.  None of this requires a smooth density or a parametric tail.

The technical difficulty is that the crash indicator jumps at $K$ while every traded payoff is continuous there.  Once the observed strikes and the threshold partition the state space into cells, the probability and first moment of each cell price every retained option exactly.  The closed cell constraints also admit limits approached from either side of $K$, although the indicator treats those sides differently.  We retain a separate coordinate for mass at $K$ to make the event convention explicit; it can be eliminated from the closed projection.  Under a strict interior condition on the quote constraints, the resulting finite system of linear inequalities projects exactly onto $\cR_K$, so the frontier can be traced by linear programming without discretizing the state space.  An adaptive algorithm recovers the polygon from finitely many support directions, and the depth frontier and its area then follow in closed form.

We apply the method to SPX quotes sampled weekly from January 2013 through August 2023, expanding the quote set from two strikes near the threshold to the complete put wing.  To measure what joint identification adds, we compare the joint set with the product of its probability and unconditional-loss intervals, intersected with $0\leq L\leq Kp$, and measure both areas after transformation to probability and depth.  In the primary sample the median area ratio is \PrimaryAreaRatio{}, corresponding to an exclusion of 36.6\% of this transformed benchmark area.  Moving from the eight strikes nearest the threshold to the complete put wing yields median within-cross-section area reductions of 5.4 to 18.2\% across the four specifications.  The shrinkage comes from requiring one distribution to fit the whole wing, not from a better fit near $K$.  A comparison with price-grid methods illustrates sensitivity near an atomic threshold: median errors decline over the tested grids, while some synthetic cases retain substantial errors at the finest resolution examined.

Section~\ref{sec:lit} places the paper in the literature.  The closest methods either discretize the state space \citep{barratt2024convex} or select a single distribution \citep{devries2026moments}; we keep the original quote intervals, characterize the exact joint region, and handle the threshold without discretization.  Sections~\ref{sec:setup}--\ref{sec:algorithm} define the identified object, derive its finite representation, and describe the reconstruction algorithm.  Sections~\ref{sec:data}--\ref{sec:benchmark} present the data, the identification gains, and the price-grid comparison.  Section~\ref{sec:conclusion} concludes, and the appendix collects proofs and computational details.

\section{Related literature}
\label{sec:lit}

\citet{merton1973theory} derives general no-arbitrage restrictions on option prices, and \citet{breeden1978prices} link the curvature of call prices across strikes to state prices.  In this tradition, nonparametric recovery selects one distribution or one arbitrage-free price surface by optimization or smoothing \citep{jackwerth1996recovering,aitsahalia1998nonparametric,fengler2009smoothing}.  Bid--ask information enters the same way.  \citet{monnier2013spectral} selects the smoothest density whose put prices lie inside the quoted intervals, \citet{wizman2026arbitrage} combine arbitrage filtering with a smoothness and entropy criterion for short-dated quotes, and \citet{qin2026marginals}, starting from discrete arbitrage-free prices, construct a marginal that reprices every input exactly.  We instead retain every distribution compatible with the intervals.  Where a selection criterion resolves observational equivalence, we measure it.

The mathematical foundation is the model-free bounds literature.  \citet{bertsimas2002relation} treat sharp option-price bounds under moment information as generalized moment problems, \citet{king2005calibrated} develop calibrated option bounds, and \citet{davis2007range} and \citet{cousot2007conditions} characterize when finitely many prices are consistent with an arbitrage-free model.  \citet{bertsimas2006option} reduce bounds for continuous piecewise-linear payoffs to linear programs, \citet{bassett1997nonparametric} derives probability bounds from finitely many option prices, and \citet{cohen2020detecting} repair option quote data by sparse linear programming.  Discontinuous payoffs behave differently.  \citet{cox2011double} document nonattainment and sensitivity to boundary mass under finite strike information, and \citet{liang2026bidask} obtain a square-root bound for upper-semicontinuous payoffs of bounded variation when the mid-marginal has a bounded density.  Our cell representation allows threshold atoms and one-sided limits in the finite-dimensional joint projection.

The closest computational work is \citet{barratt2024convex}, who bound the CDF, conditional probabilities, VaR, and CVaR over risk-neutral probabilities on a discretized price grid.  \citet{devries2026moments} construct projection estimators for general payoffs, obtain a pricing-consistent CDF from indicator payoffs, and prove a finite-sample approximation bound.  Their method returns one estimator under a chosen projection weight, whereas our target is the full feasible set under the observed intervals.  \citet{vaneekelen2023conditional} derives sharp generalized moment bounds for conditional expectations, and our object identifies the conditioning probability jointly with its shortfall numerator.  \citet{neufeld2023model} compute $\varepsilon$-optimal bounds for continuous piecewise-affine payoffs, a class that excludes digitals without prior approximation.  On the computational side, enumeration of projected polyhedra by repeated support programs is longstanding in multiobjective optimization \citep{ehrgott2012dual}, and support-function representations of convex identified sets are established in partial identification \citep{beresteanu2008asymptotic,lohne2016equivalence}.  Our contribution is the full sharp joint region under finite quote intervals, together with its analytic transformation into conditional depth and its empirical use in comparing quote sets.

A related empirical literature extracts disaster and tail information from index options \citep{backus2011disasters,bollerslev2011tails}.  \citet{vilkov2013option} fit parametric tails motivated by extreme value theory, and \citet{osullivan2018tail} estimate risk-neutral tail probabilities and tail expectations beyond the observed strike range.  Bounds on physical crash probabilities require restrictions linking physical and risk-neutral measures; \citet{martin2026forecasting} impose a power utility stochastic discount factor and bound over dependence structures consistent with option-implied marginals.  Our analysis stays entirely under $\Q$.

Finally, the fixed-threshold depth $s_K$ differs from quantile-based tail measures.  If $\Q(X=K)=0$ and $1-K$ is the relevant quantile of the loss $1-X$, then $(1-K)+s_K=1-\E_\Q[X\mid X\leq K]$ equals CVaR at the corresponding confidence level.  With mass at $K$, CVaR uses only the portion of the boundary atom needed for the prescribed tail probability, whereas our conditional depth includes the full atom \citep{rockafellar2000optimization,rockafellar2002conditional}.  Our object fixes $K$ and retains the full joint set implied by the quote intervals.

\section{Market setup and the joint tail object}
\label{sec:setup}

\subsection{Normalization and quote intervals}

Fix a date and a maturity $T$.  Let $S_T$ be the terminal SPX level, $D_{0,T}$ the discount factor, and $F_{0,T}$ the forward, and set
\[
X=\frac{S_T}{F_{0,T}}, \qquad X\geq0.
\]
Under a risk-neutral measure with the correct forward,
\[
\E_\Q[X]=1.
\]
A strike $\kappa$ in index points has normalized strike $k=\kappa/F_{0,T}$, and the time-0 call price $C(\kappa,T)$ has normalized value
\[
c_\Q(k)=\E_\Q[(X-k)^+] =\frac{C(\kappa,T)}{D_{0,T}F_{0,T}}.
\]
Let $\mathcal I$ index the retained quotes.  Quote $i\in\mathcal I$ supplies the interval
\[
\underline c_i \leq c_\Q(k_i) \leq \overline c_i.
\]
Put quotes with bid $\underline P_i$ and ask $\overline P_i$ enter through parity:
\[
[\underline c_i,\overline c_i] = \left[ \frac{\underline P_i}{D_{0,T}F_{0,T}}+1-k_i,\, \frac{\overline P_i}{D_{0,T}F_{0,T}}+1-k_i \right].
\]
We impose no midpoint, and when several retained quotes share a strike, each interval remains a separate constraint.

Let $\mathcal P(\R_+)$ denote the probability measures on $\R_+$.  The quote-compatible measures are
\[
\cM(\mathcal I)= \left\{ \Q\in\mathcal P(\R_+) \ \middle|\
    \begin{aligned}
      &\E_\Q[X]=1,\\
      &\underline c_i
      \leq\E_\Q(X-k_i)^+
      \leq\overline c_i,
      \qquad i\in\mathcal I
    \end{aligned}
\right\}.
\]
When these restrictions are jointly feasible we call the cross section feasible, and $\cM(\mathcal I)$ is the identified set of distributions under the original intervals.  Section~\ref{sec:relax} treats inconsistent cross sections.

\subsection{Probability, loss, and conditional depth}

Fix $K\in(0,1)$.  For $\Q\in\cM(\mathcal I)$, define
\begin{align}
  p_K(\Q)&=\E_\Q[\one\{X\leq K\}], \label{eq:p}\\
  L_K(\Q)&=\E_\Q[(K-X)^+], \label{eq:L}\\
  s_K(\Q)&=\frac{L_K(\Q)}{p_K(\Q)},\qquad p_K(\Q)>0. \label{eq:s}
\end{align}
The identity
\[
s_K(\Q)=K-\E_\Q[X\mid X\leq K]
\]
shows that $s_K\in[0,K]$: it is the average forward-normalized shortfall below $K$, conditional on the threshold event.  The total conditional decline from the forward is $(1-K)+s_K$, so the depth coordinate isolates the shortfall beyond the chosen threshold.

The attainable image, its closure, and the conditional depth image are
\begin{align}
  \cR_K^{\mathrm{att}}(\mathcal I)
  &=
  \{(p_K(\Q),L_K(\Q)):\Q\in\cM(\mathcal I)\},\\
  \cR_K(\mathcal I)
  &=\overline{\cR_K^{\mathrm{att}}(\mathcal I)},\\
  \cT_K(\mathcal I)
  &=
  \{(p,L/p):(p,L)\in\cR_K(\mathcal I),\ p>0\}.
\end{align}
We call the boundary of $\cR_K(\mathcal I)$ the risk-neutral crash frontier of the cross section; since every object in this paper lives under $\Q$, we usually drop the qualifier.

Every pair satisfies the universal inequalities $0\leq L_K\leq Kp_K$, a wedge in the $(p,L)$ plane.  The joint set is typically a proper subset of the product of its marginal intervals intersected with this wedge, because both coordinates must come from one quote-compatible measure.  Section~\ref{sec:data} turns this comparison into the marginal benchmark used in the empirical analysis.

\section{A finite representation with explicit threshold mass}
\label{sec:finite}

The finite representation has three ingredients: cell masses and first moments that price every retained vanilla payoff, one added coordinate that assigns threshold mass to the event, and the quote intervals that select the admissible moment vectors.  Theorem~\ref{thm:sharp} shows that together they recover the joint object exactly.

\subsection{Cell probabilities and first moments}

Let
\[
0=b_0<b_1<\cdots<b_J
\]
be the sorted union of all retained normalized strikes and the threshold $K$.  The bounded cells are $I_j=[b_j,b_{j+1})$, $j=0,\ldots,J-1$, and the terminal cell is $I_J=[b_J,\infty)$.  Define
\[
z_j=\Q(X\in I_j),\qquad m_j=\E_\Q[X\one\{X\in I_j\}],
\]
and collect them as
\[
z=(z_0,\ldots,z_J)^\top\in\R_+^{J+1},\qquad m=(m_0,\ldots,m_J)^\top\in\R_+^{J+1}.
\]
They satisfy
\begin{align}
  &z_j\geq0,\quad m_j\geq0,\quad
  \sum_{j=0}^{J}z_j=1,\quad
  \sum_{j=0}^{J}m_j=1, \label{eq:massmoment}\\
  &b_jz_j\leq m_j\leq b_{j+1}z_j,
    &&j=0,\ldots,J-1, \label{eq:boundedcell}\\
  &m_J\geq b_Jz_J. \label{eq:tailcell}
\end{align}
Every retained call payoff is affine on each cell, because its only kink is at a strike and every strike is a cell boundary, so its price depends on the measure only through $(z,m)$.  For an observed strike $k=b_\ell$,
\begin{equation}
  c_\Q(k)
  =
  \sum_{j=\ell}^{J}(m_j-kz_j).
  \label{eq:calllinear}
\end{equation}
We write $c_i(z,m)$ for the corresponding linear form of quote $i$, so every bid--ask restriction is linear in $(z,m)$.

Lemma~\ref{lem:cell} in Appendix~\ref{app:proofs} proves the converse for retained-payoff values: those implied by any vector satisfying \eqref{eq:massmoment}--\eqref{eq:tailcell} can be attained or approached by unit-mean measures.  If $z_J=0$ implies $m_J=0$ and each positive-mass cell mean $m_j/z_j$ lies in the interior of its cell, placing mass $z_j$ at that mean realizes the vector exactly.  Limits are needed at excluded cell endpoints or when the terminal cell has zero mass and positive first moment.  Thus \eqref{eq:calllinear} requires no discretization of the state space.

\begin{remark}[Order-one truncated moment problem]
\label{rem:momentproblem}
Conditions \eqref{eq:boundedcell}--\eqref{eq:tailcell} describe the closures of the order-one cell moment cones, with \eqref{eq:massmoment} imposing unit mass and mean \citep{kreinnudelman1977}.  Piecewise-affine payoffs price through cell masses and first moments alone, which is why linear constraints suffice.  Targets built from higher moments of the state would require higher cell moments and semidefinite conditions.
\end{remark}

\subsection{The threshold atom}

Let $j_K$ be the index with $K=b_{j_K}$, and introduce
\[
a_K=\Q(X=K).
\]
In the finite system, $a_K$ is the coordinate that allocates threshold mass to the event.  It satisfies
\begin{align}
  0&\leq a_K\leq z_{j_K}. \label{eq:atom1}
\end{align}
The residual mass $z_{j_K}-a_K$ lies strictly to the right of $K$ in the threshold cell, so its first moment $m_{j_K}-Ka_K$ is at least $K(z_{j_K}-a_K)$, which reduces to $m_{j_K}\geq Kz_{j_K}$ and is already imposed by \eqref{eq:boundedcell} or \eqref{eq:tailcell}.  When $j_K<J$, the same residual mass lies below $b_{j_K+1}$, which gives
\begin{align}
  m_{j_K}&\leq b_{j_K+1}z_{j_K}-(b_{j_K+1}-K)a_K.
  \label{eq:atom2}
\end{align}
When $j_K=J$, \eqref{eq:atom2} is omitted.  Together these restrictions separate mass at $K$ from residual mass to its right.  The tail coordinates are then
\begin{align}
  p_K&=\sum_{j<j_K}z_j+a_K, \label{eq:ptarget}\\
  L_K&=\sum_{j<j_K}(Kz_j-m_j). \label{eq:Ltarget}
\end{align}
The atom counts toward the probability and adds nothing to the loss, as the definitions of $p_K$ and $L_K$ require.

Let $\mathcal F_K(\mathcal I)$ denote the set of vectors $(z,m,a_K)$ that satisfy \eqref{eq:massmoment}--\eqref{eq:tailcell}, \eqref{eq:atom1}--\eqref{eq:atom2} (with \eqref{eq:atom2} omitted when $j_K=J$), and every quote interval through \eqref{eq:calllinear}.  Let $G_K$ denote the linear map \eqref{eq:ptarget}--\eqref{eq:Ltarget} from $(z,m,a_K)$ to $(p_K,L_K)$.

\begin{assumption}[Strict quote interior (a Slater condition)]
\label{ass:interior}
Every retained quote interval has positive width, and $\mathcal F_K(\mathcal I)$ contains a vector whose modeled prices lie strictly inside their respective intervals.
\end{assumption}

This constraint qualification \citep{rockafellar1970convex} ensures that closing the set of attainable cell vectors commutes with intersecting it with the closed quote intervals.  It is checked by one auxiliary linear program that maximizes the smallest relative margin between modeled prices and quote endpoints; given positive widths, the assumption holds exactly when the optimal margin is positive (Appendix~\ref{app:interior}).  A strictly interior complete quote set qualifies each of its nested subsets.

\begin{theorem}[Sharp projection with explicit threshold mass]
\label{thm:sharp}
Under Assumption~\ref{ass:interior},
\[
\cR_K(\mathcal I) = G_K\!\left(\mathcal F_K(\mathcal I)\right).
\]
\end{theorem}

Every sharp linear bound on $(p_K,L_K)$ is therefore the value of a finite-dimensional linear program.

\begin{remark}[Extremal scenarios are simple]
\label{rem:parsimony}
The measures constructed in the proof of Theorem~\ref{thm:sharp} are discrete, with at most one atom per cell, one atom at the threshold, and two auxiliary atoms.  Every point of $\cR_K(\mathcal I)$ is therefore generated, or approached, by distributions with at most $J+4$ support points, and the scenarios that trace the crash frontier are small collections of atoms.
\end{remark}

\subsection{Sensitivity at the threshold}

Moving mass from $K-\varepsilon$ to $K+\varepsilon$ changes vanilla prices continuously but changes $p_K$ discontinuously.  A price-grid method, which replaces $\Q$ by probabilities on fixed nodes, inherits this problem: it can match every continuous payoff closely and still assign threshold mass to the wrong side of the event.

\begin{proposition}[Threshold nonuniformity]
\label{prop:grid}
Fix $\alpha\in(0,1)$.  For every sufficiently small $\varepsilon>0$ there are two unit-mean distributions on $\R_+$, one placing mass $\alpha$ at $K-\varepsilon$ and the other at $K+\varepsilon$, whose call prices differ by at most $4\alpha\varepsilon$ at every strike.  Their values of $\Q(X\leq K)$ differ by $\alpha$.
\end{proposition}

Thus uniformly small changes in call prices need not imply small changes in the threshold probability when atoms are allowed.  This statement concerns stability across distributions; it does not establish failure of grid bounds to converge for a fixed collection of positive-width quote intervals.  Section~\ref{sec:benchmark} examines finite-grid errors under the stated quote budgets.

\subsection{Inconsistent cross sections}
\label{sec:relax}

Some cross sections admit no distribution consistent with every quote interval.  For these cases we relax each interval by nonnegative slacks, penalize the slacks through fixed positive weights, and impose a common budget on the weighted slack sum.  The budget is fixed before the main analysis and shared by every nested quote set within a cross section, so the resulting regions remain nested.  The relaxed system is again a finite linear system, and the reconstruction of Section~\ref{sec:algorithm} applies without change.  Appendix~\ref{app:relax} states the relaxed system, the weights, and the relaxed region under the budget, written $\cR_{K,\bar\delta}(\mathcal I)$.  Sharp results under the original intervals use feasible cases satisfying Assumption~\ref{ass:interior}.  The sensitivity summaries in Table~\ref{tab:robustness} also include common-budget outer relaxations, as stated in its notes.

\section{Frontier reconstruction}
\label{sec:algorithm}

\subsection{Support programs}

Collect the decision variables in $y=(z^\top,m^\top,a_K)^\top$ and write $\cY$ for the feasible polytope: $\cY=\mathcal F_K(\mathcal I)$ under the original intervals, and its slack-augmented analogue under the common budget (Appendix~\ref{app:relax}), in which case $G_K$ acts on the $(z,m,a_K)$ block of $y$.  The mass and first-moment normalizations, together with the budget on any slacks, make $\cY$ bounded, so its projection is a compact polygon.  For a direction $u=(u_p,u_L)\in\R^2$, the support program is the linear program
\begin{equation}
  h(u)
  =
  \max_{y\in\cY}\ u^\top G_Ky.
  \label{eq:support}
\end{equation}

\begin{corollary}[Sharp portfolio bounds]
\label{cor:portfolio}
Let Assumption~\ref{ass:interior} hold and let $\cY=\mathcal F_K(\mathcal I)$.  For every $u=(u_p,u_L)\in\R^2$,
\[
h(u) = \sup_{\Q\in\cM(\mathcal I)} \E_\Q\!\left[u_p\one\{X\leq K\}+u_L(K-X)^+\right],
\]
and the corresponding sharp lower value is $-h(-u)$.
\end{corollary}

The corollary restates Theorem~\ref{thm:sharp} through support functions, which coincide for a convex set and its closure.  Every portfolio of the digital and put payoffs is priced sharply by one linear program, and once the polygon is known no further program is needed.  Appendix~\ref{app:algorithm} describes the vertex selection used for polygon enumeration and the primal--dual accuracy checks.

\subsection{Adaptive hull algorithm}

The hull procedure adapts standard support-function and outer-approximation methods \citep{ehrgott2012dual,lohne2016equivalence}.  Let $V_t$ be the support points collected after iteration $t$ and $H_t=\operatorname{conv}(V_t)$.  The reconstruction initializes $V_0$ from eight evenly spaced directions, queries the outward normal of each current edge, and adds any support point outside the current hull.  In exact arithmetic this process recovers a point, a line segment, or the full polygon after finitely many support solves; Proposition~\ref{prop:hull} in Appendix~\ref{app:proofs} gives the formal result, and Appendix~\ref{app:algorithm} the implementation.  Dual support halfspaces also define an outer polygon, and the inner--outer area difference measures reconstruction error.

\subsection{The conditional depth image}

For $p>0$, define $\phi(p,L)=(p,L/p)$.  The map preserves the first coordinate but not convexity, so the depth image is recovered from the polygon edge by edge.  To measure joint ambiguity away from zero probability, fix a floor $p_0>0$.  For a nonempty compact convex $\cC\subset\R^2$, let $\cC(p)=\{L:(p,L)\in\cC\}$, and when the section is nonempty write $L_{\cC}^+(p)=\max\cC(p)$ and $L_{\cC}^-(p)=\min\cC(p)$.  Define
\begin{equation}
  \mathcal A_{p_0}(\cC)
  =
  \int_{\{p\geq p_0:\cC(p)\ne\varnothing\}}
  \frac{L_{\cC}^+(p)-L_{\cC}^-(p)}{p}\,dp,
  \label{eq:area}
\end{equation}
the transformed area of $\cC$.

\begin{proposition}[Depth image and transformed area]
\label{prop:depth}
On $\{p>0\}$, $\phi$ is a diffeomorphism with $|\det D\phi(p,L)|=1/p$.  A nonvertical polygon edge with equation $L=\alpha_e+\beta_e p$ maps to the arc $s(p)=\beta_e+\alpha_e/p$, and a vertical edge maps to a vertical segment.  For every nonempty compact convex $\cC$ and every $p_0>0$,
\[
\mathcal A_{p_0}(\cC) = \operatorname{Leb}_2\!\left(\phi(\cC\cap\{p\geq p_0\})\right),
\]
where $\operatorname{Leb}_2$ is planar Lebesgue measure.
\end{proposition}

The edge formulas are direct substitution, the area identity is the change-of-variables formula, and for polygonal $\cC$ the integral in \eqref{eq:area} has the closed form given in Appendix~\ref{app:algorithm}.  On any domain $p\geq p_0>0$, the relative boundary of $\cT_K(\mathcal I)$ is therefore a finite collection of vertical segments and arcs of the form $\beta_e+\alpha_e/p$.  The frontier itself is defined for every $p>0$; the floor $p_0$ enters only the area summary.  Because area depends on the chosen coordinates, the full frontier remains the primary identified object.

\subsection{Dual certificates}

The dual of \eqref{eq:support} attaches multipliers to the bid, ask, mass, first-moment, cell, and threshold-atom constraints, and strong duality gives the same support value in both programs (Appendix~\ref{app:dual}).  For each exposed boundary point, an optimal dual solution decomposes the bound into cash and forward terms, a signed portfolio of retained quote endpoints, and nonnegative structural cell terms.  Quote constraints with nonzero multipliers identify which intervals produce the bound; dual optima need not be unique, and the certificate reports one optimal solution.  We call this decomposition a dual certificate.

Read as a portfolio, the certificate is a static super-replication of the target payoff: cash and forward positions plus retained quotes executed at bid or ask, whose combined payoff dominates the digital and put target state by state, with the structural terms certifying the domination on each cell \citep{davis2007range,cox2011double}.  The upper bound is the cost of the super-replicating portfolio; a sub-replicating portfolio certifies the lower bound.  For relaxed cross sections the dual carries an additional budget multiplier, so the portfolio reading is exact only under the original intervals.

\section{Data and empirical design}
\label{sec:data}

\subsection{SPX cross sections}

The application uses end-of-day SPX call and put bid--ask quotes from OptionMetrics, from 2 January 2013 through 31 August 2023.  Normalization uses the SPX level and a forward estimate for the selected date and maturity, and robustness checks apply multiplicative shifts of $\pm20$ basis points to the forward.

The sampling rule selects one eligible cross section per calendar week, preferring midweek dates.  Each selected date must offer a short expiry of 7--21 days closest to 14 and a medium expiry of 22--60 days closest to 45, with at least 18 distinct put strikes in log-moneyness $[-0.35,0.05]$.  The sample has \SampleWeeks{} dates from 2 January 2013 to 30 August 2023, 98.92\% of them Wednesdays.  Appendix~\ref{app:empirical} records the weekday fallback order and the tie rules.

We retain positive strikes, nonnegative bids, positive asks, noncrossed bid--ask pairs, positive trading volume, and expiries between 7 and 120 days.  The primary quote sets use puts converted by parity, and a call--put cross section is reserved for robustness.  The complete put wing comprises all retained puts with log-moneyness in $[-0.35,0.05]$ for the selected date and expiry, and multiple retained observations at the same strike remain separate constraints.

\subsection{Prespecified estimands and quote sets}

The main thresholds are $K\in\{0.85,0.90\}$, forward-normalized declines of 15\% and 10\%.  Within each cross section we compare four objects:
\begin{enumerate}[label=(\roman*),leftmargin=2.4em]
\item the marginal benchmark: the wedge-adjusted product of separate marginal bounds, defined below;
\item the put pair: the nearest puts at or below and at or above $K$;
\item the local eight: the put pair plus the six nearest remaining put strikes;
\item the complete put wing $\mathcal I^C$: every eligible put in the cross section.
\end{enumerate}
Tie rules for the put pair are recorded in Appendix~\ref{app:empirical}, and the slack weights used in the relaxation analysis in Appendix~\ref{app:relax}.

For a nonempty compact convex $\cC\subset\R^2$, the wedge-adjusted marginal benchmark is
\[
\cB_K(\cC)= \left( [\min_{\cC}p,\max_{\cC}p]\times [\min_{\cC}L,\max_{\cC}L] \right) \cap\{(p,L):0\leq L\leq Kp\},
\]
with extrema over $(p,L)\in\cC$, so that $\cR_K(\mathcal I)\subseteq\cB_K(\cR_K(\mathcal I))$.  The reported area ratio is
\[
\frac{\mathcal A_{p_0}(\cC)} {\mathcal A_{p_0}(\cB_K(\cC))},
\]
the share of the transformed benchmark area that survives joint identification.  By Proposition~\ref{prop:depth}, both areas are measured in probability and depth coordinates.  The benchmark is constructed from the $p$ and $L$ marginals; its transformed image generally differs from the rectangle formed by separate sharp bounds on $p$ and $s$.  The excluded share is an area fraction in these coordinates, not a probability assigned to scenarios.  Smaller ratios indicate tighter joint identification, and a ratio of one means that joint restrictions add nothing beyond the marginal intervals and the universal wedge on $p\geq p_0$.  The ratio is reported for two-dimensional joint sets with positive benchmark area; point and line sets are classified as degenerate.  The main specification uses $p_0=10^{-4}$, with $10^{-5}$ and $10^{-3}$ as robustness values.

For a quote-based baseline $A$, write $\cR_K^A=\cR_K(\mathcal I^A)$ under the original intervals and $\cR_{K,\bar\delta}^A=\cR_{K,\bar\delta}(\mathcal I^A)$ under the common budget.  Because the quote sets are nested within each cross section, $\mathcal I^A\subseteq\mathcal I^{C}$ implies $\cR_K^{C}\subseteq\cR_K^A$ and $\cR_{K,\bar\delta}^{C}\subseteq\cR_{K,\bar\delta}^A$.  When the baseline area is positive, the identification gain from set $A$ to the complete put wing is
\[
\Delta_A = 1-\frac{\mathcal A_{p_0}(\cR_K^{C})} {\mathcal A_{p_0}(\cR_K^{A})},
\]
the fractional contraction in transformed area from the additional quotes.  It satisfies $0\leq\Delta_A\leq1$ in exact arithmetic.  Sensitivity calculations use the same formula with $\cR_{K,\bar\delta}^A$ in place of $\cR_K^A$.

Feasibility under the original bid--ask constraints is the primary criterion and holds in \HardShare{} of the \MainCases{} cases with the complete put wing.  Assumption~\ref{ass:interior} holds in \InteriorQualifiedCases{} of these cases.  Cross sections with ambiguous contract identity are excluded, leaving \QualifiedCases{} cases in the primary empirical sample.  Sixteen feasible cases fail Assumption~\ref{ass:interior}, and 136 cases are infeasible without relaxation and enter the common-budget sensitivity analysis.

\subsection{Sampling variation and robustness}

We form percentile intervals for median identification gains from 2,000 draws of a circular block bootstrap.  Blocks contain eight chronologically ordered dates, the cube-root rule applied to the number of sampled dates in each primary cell.  Resampling dates rather than individual quotes preserves market cross sections and allows short-range serial dependence \citep{lahiri2003resampling}.

Robustness checks target four design choices: the probability floor, forward and spread inputs, strike and duplicate selection, and the maturity and threshold grid.  Probability-floor results reuse the reconstructed polygons, and the remaining checks use selected date subsamples.  Appendix~\ref{app:empirical} summarizes the specifications.

\section{Empirical results}
\label{sec:results}

\subsection{Main identification gains}

Table~\ref{tab:main} summarizes identification gains in the primary sample.  Across the four combinations of horizon and threshold, the complete put wing contracts the area of the local eight by 5.35\% to 18.16\% at the median.  Quotes beyond the eight nearest strikes therefore eliminate between one-twentieth and one-fifth of the joint ambiguity that those strikes leave open.

The contraction is global rather than local in origin.  Quotes away from $K$ restrict where a common risk-neutral distribution can place compensating probability and first moment while keeping every modeled price inside its interval and satisfying $\E_\Q[X]=1$.  A reallocation that fits the nearby strikes can therefore fail elsewhere on the wing.  The complete put wing disciplines the feasible probability--depth combinations rather than refining the interpolation at $K$.

\begin{table}[t]
\centering
\caption{Incremental identification from the complete put wing}
\label{tab:main}
\begin{tabular}{llrrrr}
\toprule
Horizon & $K$ & \multicolumn{2}{c}{Put pair} &
\multicolumn{2}{c}{Local eight}\\
\cmidrule(lr){3-4}\cmidrule(lr){5-6}
& & Median & 95\% interval & Median & 95\% interval\\
\midrule
Short  & 0.85 & 74.33\% & [70.15,77.59] &  5.35\% & [0.52,13.92]\\
Short  & 0.90 & 83.10\% & [81.42,84.27] & 15.03\% & [10.38,22.75]\\
Medium & 0.85 & 84.75\% & [83.28,86.25] &  8.12\% & [2.73,13.71]\\
Medium & 0.90 & 89.03\% & [87.63,90.39] & 18.16\% & [13.19,21.56]\\
\bottomrule
\end{tabular}
\begin{minipage}{0.94\textwidth}
\footnotesize\emph{Notes:} Each entry is
$1-\mathcal A_{p_0}(\cR_K^{C})/ \mathcal A_{p_0}(\cR_K^{\mathrm{baseline}})$, with $p_0=10^{-4}$.  Larger entries indicate greater contraction relative to the stated baseline.  For each threshold, there are 535 primary short-horizon cases and 448 primary medium-horizon cases, giving 1,966 cases in total.  Intervals use 2,000 circular block bootstrap draws with eight chronologically ordered dates satisfying the primary-sample conditions per block.  Point estimates use the midpoint of the numerical area interval, and the intervals summarize sampling variation only; numerical accuracy is verified separately (Appendix~\ref{app:empirical}).
\end{minipage}
\end{table}

Across primary-sample cases with the complete put wing, the median ratio of joint area to wedge-adjusted marginal benchmark area is \PrimaryAreaRatio{}.  The joint restriction therefore excludes a median 36.6\% of the transformed benchmark area defined in Section~\ref{sec:data}.

Median probability widths under the complete put wing are 0.0045 and 0.011 at the short horizon and 0.019 and 0.037 at the medium horizon, for $K=0.85$ and $0.90$.  Marginal probability uncertainty therefore remains visible even where joint restrictions sharply reduce the probability--depth region.

Figure~\ref{fig:frontier} displays the representative case used for the dual certificate, the short-horizon cross section dated 5 July 2017 with maturity 19 July 2017 and $K=0.90$.  It satisfies the original quote constraints and Assumption~\ref{ass:interior}.

\begin{figure}[t]
\centering
\includegraphics[width=\textwidth]{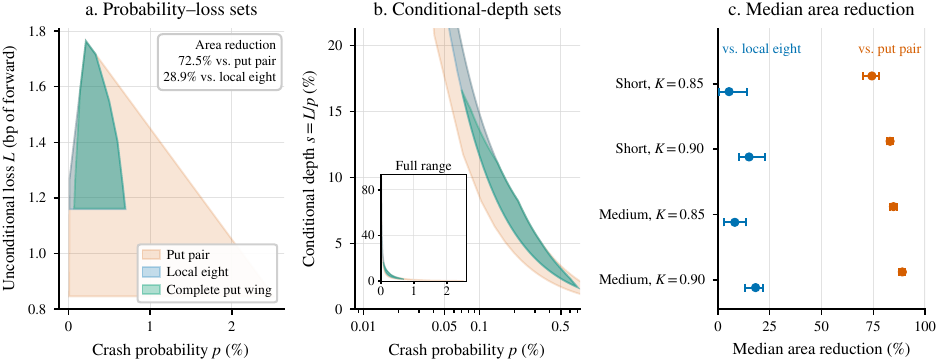}
\caption{Sharp joint frontiers and primary-sample identification gains.  The representative case is also used for the dual certificate in Figure~\ref{fig:dual}.  Filled regions in panels (a) and (b) correspond to the put pair, local eight, and complete put wing.  Panel (a) plots the probability--loss set with loss in basis points of forward.  Panel (b) zooms the depth range on a log-$p$ axis and its inset restores the full range.  Panel (c) reports primary-sample median area reductions and circular block bootstrap intervals for the four main specifications.}
\label{fig:frontier}
\end{figure}

\subsection{Stability and robustness}

Appendix Table~\ref{tab:subsample} reports subperiod summaries and results that exclude the calibration dates, both under the primary sample definition.  The median area ratio varies little across subperiods.

Median identification gains remain positive under prespecified perturbations to forwards, quotes, strike coverage, horizons, and thresholds.  Appendix Table~\ref{tab:robustness} reports the range of median ratios for each specification.  The denominator is itself a moving benchmark, so doubled spreads can lower the ratio even as absolute probability widths grow.

\section{Price-grid benchmark and dual evidence}
\label{sec:benchmark}

\subsection{Median errors and atomic cases on finite grids}

The benchmark follows the price-grid approach of \citet{barratt2024convex}.  We repeat the support analysis on price grids with 51, 101, 201, and 401 nodes spanning $[0,2]$ or $[0,3]$ (support caps 2 and 3), with the threshold either inserted as a node or centered in a cell.  Every grid solve is compared with a continuous-state solve in the same 64 directions and under the same quote loss budget, so support errors compare the two state representations at a common relaxation level.  Each specification contains 36 cases: 32 empirical combinations of cross section and threshold, and four synthetic atomic cases.  Appendix~\ref{app:empirical} details the grid construction, the shared budget, and the error definitions.

The budget required for grid feasibility is not negligible at 51 nodes, with median increments over the continuous-state minimum of $(1.22$--$1.63)\times10^{-4}$ at cap 2, and it falls to the numerical floor of $10^{-10}$ at 201 and 401 nodes.  Because the budget varies with the grid, changes across grids combine refinement with changes in the relaxation level.  Table~\ref{tab:grid} reports the values.

Figure~\ref{fig:boyd} separates typical from worst-case error on the tested grids.  Median error declines strongly and is numerically zero with $K$ inserted and 401 nodes.  The maximum remains near 0.15 with $K$ inserted, while centering $K$ in a cell gives values near 0.075.  Both patterns arise in the synthetic 0.15 threshold-atom design and are consistent with the sensitivity in Proposition~\ref{prop:grid}.  An inserted node assigns its mass to $\{X\leq K\}$, whereas the closed continuous-state system also represents limits from immediately above $K$.  These finite-grid comparisons use grid-dependent quote budgets and do not establish an asymptotic error floor for fixed quotes.  On the real SPX cross sections at support cap 2, the maximum error falls from 0.0751 to 0.0100 for centered grids and from 0.0784 to 0.0138 with $K$ inserted.

\begin{figure}[t]
\centering
\includegraphics[width=\textwidth]{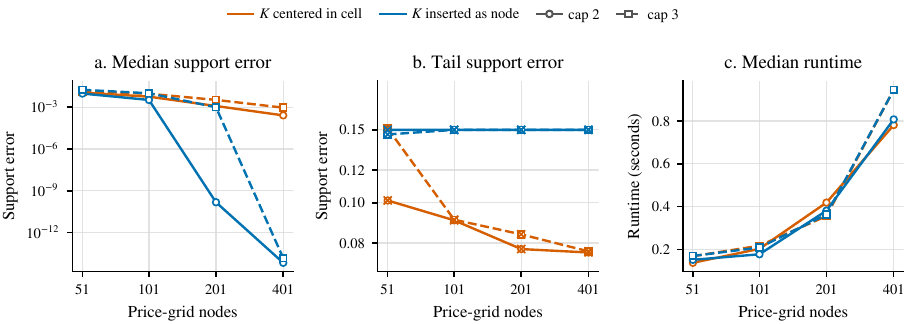}
\caption{Price-grid accuracy and computation time.  Panel (a) first takes the median support error over directions within each case, then reports the median of those values across cases.  In panel (b), lines report the 95th percentile of each case's maximum support error over directions, and crosses report sample maxima.  Panel (c) reports median computation time.  Synthetic threshold-atom designs generate the nearly flat worst-case errors over the tested resolutions.  Every plotted grid is compared with a continuous-state solve under the same quote loss budget for that grid.}
\label{fig:boyd}
\end{figure}

Computation time rises with grid density, as panel (c) shows.  The two methods address different computational objectives: the continuous-state method reconstructs an entire frontier, whereas the benchmark evaluates a fixed set of directions on each grid.

\subsection{Primal--dual verification and a dual certificate}

For the representative case, primal--dual gaps and KKT residuals are below $2\times10^{-14}$ in the directions that maximize probability, loss, and the probability--loss trade-off.  At a $10^{-8}$ multiplier threshold, the certificate for each direction puts weight on two quote endpoints, and the remaining terms enforce the mass, cell-moment, and threshold-atom restrictions.  Figure~\ref{fig:dual} displays an optimal trade-off dual certificate.

\begin{figure}[t]
\centering
\includegraphics[width=\textwidth]{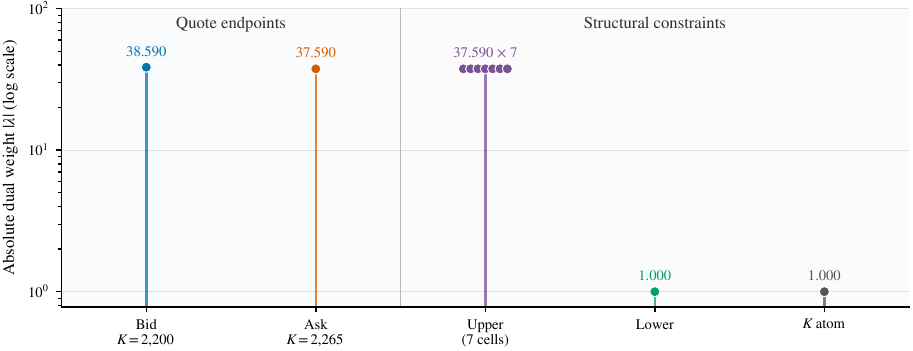}
\caption{Representative trade-off dual certificate for direction $u^\star=(1,-1)$ in the representative frontier, shown at its original scale.  Axis labels identify each quote by strike and endpoint.  Only quote weights above the $10^{-8}$ threshold are displayed.}
\label{fig:dual}
\end{figure}

\section{Conclusion}
\label{sec:conclusion}

Option-implied crash probability and conditional depth are two readings of the same put value, tied together by the risk-neutral distribution that generates both.  Marginal intervals lose that tie, and the risk-neutral crash frontier is built to keep it.  Under a verifiable interior condition the joint set is the exact projection of a finite system of linear inequalities with an explicit coordinate for threshold mass, linear programs recover its boundary, and every point on the boundary comes with a dual certificate that names the individual quotes behind it.

In more than ten years of weekly SPX cross sections, the complete put wing yields median transformed-area reductions of 5.4--18.2\% relative to the eight nearest strikes and 74.3--89.0\% relative to the nearest pair across the four specifications.  Strikes far from the threshold therefore carry real information about which joint tail scenarios remain compatible with the market.  On the tested grids, median errors fall with refinement while atomic threshold cases remain sensitive to node placement and quote tolerance.  The continuous-state frontier represents the relevant one-sided limits directly and turns a finite option cross section into valuation bounds, feasible stress scenarios, comparisons across strike sets, and certificates that trace each bound back to observed quotes.

The analysis prices one maturity at a time and stays under the risk-neutral measure.  Partitions that couple several maturities, and preference restrictions that map the frontier into bounds on physical crash probabilities, are natural extensions.

\appendix
\renewcommand{\thetable}{A\arabic{table}}
\setcounter{table}{0}

\section{Appendix: Proofs, computation, and supplementary evidence}
\label{app:supp}

\subsection{Proofs and auxiliary results}
\label{app:proofs}

\subsubsection{Cell representation}

\begin{lemma}[Cell-moment representation for retained payoffs]
\label{lem:cell}
Let each retained payoff $g_\ell$ be continuous and affine on every cell, $g_\ell(x)=\alpha_{\ell j}+\beta_{\ell j}x$ for $x\in I_j$.  The cell vector $(z,m)$ of any unit-mean probability measure $\Q$ satisfies \eqref{eq:massmoment}--\eqref{eq:tailcell}, and
\[
\E_\Q[g_\ell(X)]=\sum_{j=0}^J(\alpha_{\ell j}z_j+\beta_{\ell j}m_j)\qquad\text{for every }\ell.
\]
Conversely, for every $(z,m)$ satisfying \eqref{eq:massmoment}--\eqref{eq:tailcell} there are unit-mean probability measures $\Q_n$ with
\[
\E_{\Q_n}[g_\ell(X)] \longrightarrow \sum_{j=0}^J(\alpha_{\ell j}z_j+\beta_{\ell j}m_j) \qquad\text{for every }\ell.
\]
\end{lemma}

\begin{proof}
Integrating $b_j\leq x<b_{j+1}$ over a bounded cell and $x\geq b_J$ over the terminal cell gives \eqref{eq:boundedcell}--\eqref{eq:tailcell}, and \eqref{eq:massmoment} is unit mass and unit mean.  Since $g_\ell$ is affine on $I_j$,
\[
\E_\Q[g_\ell(X)\one\{X\in I_j\}] =\alpha_{\ell j}z_j+\beta_{\ell j}m_j,
\]
and summing over $j$ gives the pricing formula.

For the converse, the closures of the attainable mass and first-moment cones on a bounded cell and on the terminal cell are, respectively,
\[
\{(z,m):z\geq0,\ b_jz\leq m\leq b_{j+1}z\}, \qquad \{(z,m):z\geq0,\ m\geq b_Jz\}.
\]
For $z>0$, mass $z$ at $m/z$ attains $(z,m)$ unless $m/z=b_{j+1}$, the excluded right endpoint of a bounded cell, which is approached from the left.  On the terminal cell, $(0,m)$ with $m>0$ is approached by mass $1/n$ at $nm$.  Applying these constructions cellwise gives finite measures $\mu_n$ whose payoff vectors converge to the prescribed vector and whose masses $M_n$ and first moments $U_n$ converge to one.  Let $X_n$ have law $\mu_n/M_n$ and set $\widehat X_n=(M_n/U_n)X_n$, so that $\E[\widehat X_n]=1$.  Each $g_\ell$ is Lipschitz, hence
\[
\left|\E g_\ell(\widehat X_n)-\E g_\ell(X_n)\right| \leq \operatorname{Lip}(g_\ell) \left|\frac{M_n}{U_n}-1\right|\frac{U_n}{M_n} \longrightarrow0.
\]
Since $M_n\to1$, the payoff vectors of the laws of $X_n$, and therefore of $\widehat X_n$, converge to the prescribed vector.
\end{proof}

\subsubsection{Proof of Theorem~\ref{thm:sharp}}

\begin{proof}
\emph{Forward inclusion.} For $\Q\in\cM(\mathcal I)$, set $a_K=\Q(X=K)$.  The residual threshold-cell pair is $(z_{j_K}-a_K,m_{j_K}-Ka_K)$.  It satisfies \eqref{eq:atom1}, and when $j_K<J$ its upper cell bound is
\[
m_{j_K}-Ka_K\leq b_{j_K+1}(z_{j_K}-a_K),
\]
which is \eqref{eq:atom2}.  Splitting the cells below $K$ from the atom at $K$ gives \eqref{eq:ptarget}--\eqref{eq:Ltarget}.  Therefore $\cR_K^{\mathrm{att}}(\mathcal I)\subseteq G_K(\mathcal F_K(\mathcal I))$.  The latter is the linear image of a compact polyhedron and is closed, so
\[
\cR_K(\mathcal I)\subseteq G_K(\mathcal F_K(\mathcal I)).
\]

\emph{Reverse inclusion: cell realization.} Fix $y=(z^\top,m^\top,a_K)^\top\in\mathcal F_K(\mathcal I)$.  Apply the cellwise constructions of Lemma~\ref{lem:cell}, before global normalization, to every nonthreshold pair and to $(z_{j_K}-a_K,m_{j_K}-Ka_K)$, and place mass $a_K$ at $K$.  The structural constraints put the residual pair in the closure of the moment cone of $(K,b_{j_K+1})$ or $(K,\infty)$.  A residual mean equal to $K$ is approached from the right, which keeps that mass outside $\{X\leq K\}$ while all continuous-payoff values converge, and Lemma~\ref{lem:cell} handles the remaining cell boundaries.  This gives finite nonnegative measures $\widetilde\mu_n$ with mass $1+A_n$ and first moment $1+B_n$, where $A_n,B_n\to0$.

\emph{Mass and mean correction.} Fix $x_{\rm buf}>\max\{b_J,1\}$ and the unit-mass, unit-mean buffer
\[
\mu_{\rm buf} =\left(1-\frac1{x_{\rm buf}}\right)\delta_0 +\frac1{x_{\rm buf}}\delta_{x_{\rm buf}}.
\]
For sufficiently large $n$, set $\varepsilon_n=(|A_n|+|B_n|)^{1/2}+n^{-1}<1$ and $\bar\mu_n=(1-\varepsilon_n)\widetilde\mu_n+ \varepsilon_n\mu_{\rm buf}$.  If $\bar A_n$ and $\bar B_n$ are its mass and moment errors, set
\[
\mu_n =\bar\mu_n + \left(-\bar A_n+\frac{\bar B_n}{x_{\rm buf}}\right)\delta_0 -\frac{\bar B_n}{x_{\rm buf}}\delta_{x_{\rm buf}}.
\]
Then $\mu_n$ has exactly unit mass and unit mean.  The correction has total absolute size $O(|A_n|+|B_n|)$, which is $o(\varepsilon_n)$ by the choice of $\varepsilon_n$, while the buffer atoms at $0$ and $x_{\rm buf}$ have mass of order $\varepsilon_n$; hence $\mu_n$ is nonnegative for all large $n$.  The correction changes the target coordinates and every retained vanilla value by $o(1)$, preserves support in $\R_+$, and does not move the atom at $K$.  Thus the projections converge to $G_K(y)$.

\emph{Quote compatibility.} Let $y^\circ$ be the vector from Assumption~\ref{ass:interior} whose modeled prices lie strictly inside all quote intervals, and put $y_\theta=(1-\theta)y+\theta y^\circ$.  For every $0<\theta\leq1$, its modeled prices have positive distance from all quote endpoints, so the preceding approximation can be chosen inside every quote interval.  A diagonal sequence with $\theta\downarrow0$ is therefore quote-compatible and has projection converging to $G_K(y)$.  Hence
\[
G_K(\mathcal F_K(\mathcal I))\subseteq\cR_K(\mathcal I),
\]
which proves the equality.
\end{proof}

\paragraph{Eliminating the threshold coordinate.}
\label{app:atom-elimination}
The explicit atom is useful for recording the event convention, but is redundant in the closed projection.  Put $j=j_K$.  For any $(z,m,a_K)\in\mathcal F_K(\mathcal I)$, set
\[
\begin{aligned}
z'_{j-1}&=z_{j-1}+a_K,& m'_{j-1}&=m_{j-1}+Ka_K,\\
z'_j&=z_j-a_K,& m'_j&=m_j-Ka_K,& a'_K&=0,
\end{aligned}
\]
and leave all other coordinates unchanged.  The left cell gains mass at its closed upper endpoint $K$, and the right cell retains precisely its residual moments.  Hence all structural inequalities remain valid, including when $j=J$.  Total mass and first moment are unchanged, and continuity of the retained payoffs at $K$ preserves every quote value.  Equations~\eqref{eq:ptarget}--\eqref{eq:Ltarget} give the same $(p_K,L_K)$, so
\[
G_K(\mathcal F_K(\mathcal I))
=G_K\bigl(\mathcal F_K(\mathcal I)\cap\{a_K=0\}\bigr).
\]
This is an identity for the closed finite system; the transformed coordinates need not be the half-open-cell coordinates of the same measure.

\subsubsection{Finite termination of adaptive hull reconstruction}

\begin{proposition}[Finite termination of adaptive hull reconstruction]
\label{prop:hull}
Let $P\subset\R^2$ be a nonempty compact convex polygon, possibly a point or a segment.  For $u\in\R^2\setminus\{0\}$, define
\[
h_P(u)=\max_{x\in P}u^\top x, \qquad x(u)\in\arg\max_{x\in P}u^\top x,
\]
with $x(u)$ selected as an endpoint of the exposed face.  At iteration $t$, write $H_t=\operatorname{conv}(V_t)$.  For an edge $e$ of $H_t$, let $n_e$ be its outward unit normal and $v_e$ either endpoint, and define
\[
\gamma_t(e)=h_P(n_e)-n_e^\top v_e.
\]
In exact arithmetic, the dimension probes of Appendix~\ref{app:algorithm} identify a point or a line segment.  If $P$ is two-dimensional, repeatedly adding $x(n_e)$ whenever $\gamma_t(e)>0$ terminates after finitely many support queries with $H_t=P$.
\end{proposition}

\begin{proof}
Every collected support point belongs to $P$, so $H_t\subseteq P$.  If all initial points equal $q$, the opposing coordinate supports imply $x_i=q_i$ for $i=1,2$ and every $x\in P$, hence $P=\{q\}$.  If the initial points span a line $\ell=\{x:n^\top x=c\}$, the probes in directions $n$ and $-n$ establish $P\subseteq\ell$ exactly when their support values are $c$ and $-c$; otherwise one probe returns a noncollinear point.  In the collinear case, the opposing coordinate probes select the endpoints of a nontrivial segment.

Suppose now that $P$ is two-dimensional.  If $\gamma_t(e)>0$, then
\[
n_e^\top x(n_e)=h_P(n_e)>n_e^\top v_e,
\]
so $x(n_e)\notin H_t$, and endpoint selection makes $x(n_e)$ a vertex of $P$.  Every expansion therefore adds a previously absent vertex.  Since $P$ has finitely many vertices and each intermediate hull has finitely many edges, only finitely many support queries are made.

At termination every edge gap is zero, and therefore
\[
H_t\subseteq P \subseteq\bigcap_{e\in\mathcal E(H_t)} \{x:n_e^\top x\leq h_P(n_e)\} =\bigcap_{e\in\mathcal E(H_t)} \{x:n_e^\top x\leq n_e^\top v_e\} =H_t.
\]
Hence $H_t=P$.
\end{proof}

\subsubsection{Proof of Proposition~\ref{prop:grid}}

\begin{proof}
Choose $0<\beta<1-\alpha$ and, for sufficiently small $\varepsilon>0$, set
\[
\Q_\pm=(1-\alpha-\beta)\delta_0 +\alpha\delta_{K\pm\varepsilon}+\beta\delta_{x_\pm}, \qquad x_\pm=\frac{1-\alpha(K\pm\varepsilon)}{\beta}.
\]
Because $K<1$, these are nonnegative unit-mass, unit-mean distributions with $x_\pm>K$ for all sufficiently small $\varepsilon$, and their probabilities of $\{X\leq K\}$ differ by $\alpha$.

Couple the threshold atoms with each other and the high atoms with each other.  The transport costs are $2\alpha\varepsilon$ and
\[
\beta|x_+-x_-|=2\alpha\varepsilon,
\]
so the Wasserstein-1 distance satisfies $W_1(\Q_+,\Q_-)\leq4\alpha\varepsilon$.  Every call payoff is 1-Lipschitz, so each call-price difference is at most $4\alpha\varepsilon$.  The vanilla price vectors converge while the threshold probabilities remain $\alpha$ apart.
\end{proof}

\subsection{Computational method and validation}
\label{app:algorithm}

\subsubsection{Verification of the strict interior condition}
\label{app:interior}

Let $\mathcal F_K^0$ collect the structural cell and threshold restrictions \eqref{eq:massmoment}--\eqref{eq:tailcell} and \eqref{eq:atom1}--\eqref{eq:atom2} before the quote inequalities.  Solve
\[
\begin{aligned}
 \eta^\star
 &=
 \max_{z,m,a_K,\eta}\quad \eta\\
 \text{subject to}\quad
 &(z,m,a_K)\in\mathcal F_K^0,\qquad 0\leq\eta\leq\tfrac12,\\
 &\underline c_i+\eta(\overline c_i-\underline c_i)
 \leq c_i(z,m),\qquad i\in\mathcal I,\\
 &c_i(z,m)
 \leq\overline c_i-\eta(\overline c_i-\underline c_i),
 \qquad i\in\mathcal I.
\end{aligned}
\]
Given positive widths for all retained quote intervals, Assumption~\ref{ass:interior} is equivalent to $\eta^\star>0$.  The numerical acceptance threshold is stated below.

\subsubsection{Relaxed system and common budget}
\label{app:relax}

For inconsistent cross sections, introduce interval slacks $\xi_i^-,\xi_i^+\geq0$ and replace each quote constraint by
\[
\underline c_i-\xi_i^- \leq c_i(z,m) \leq \overline c_i+\xi_i^+.
\]
Call $(z,m,a_K,\xi^-,\xi^+)$ feasible for the relaxed system if $(z,m,a_K)$ satisfies \eqref{eq:massmoment}--\eqref{eq:tailcell} and \eqref{eq:atom1}--\eqref{eq:atom2} (with \eqref{eq:atom2} omitted when $j_K=J$), the slacks are nonnegative, and the displayed inequalities hold for every $i\in\mathcal I$.

Weights are set on the complete put wing $\mathcal I^C$.  If it has $N_C$ distinct strikes and $n_C(k_i)$ observations at strike $k_i$,
\[
w_i=\frac{1}{N_C\,n_C(k_i)},\qquad \sum_{i\in\mathcal I^C}w_i=1,
\]
so each distinct strike receives equal total weight, divided equally across observations at that strike.  Each nested set inherits the weights on its retained constraints, which avoids renormalizing common constraints.  The minimum weighted quote loss is
\[
\begin{aligned}
  \delta^\star(\mathcal I)
  &=
  \min_{z,m,a_K,\xi^-,\xi^+}
  \quad \sum_{i\in\mathcal I}w_i(\xi_i^-+\xi_i^+)\\
  \text{subject to}\quad
  &(z,m,a_K,\xi^-,\xi^+)
  \text{ is feasible for the relaxed system}.
\end{aligned}
\]
Under the prespecified common budget $\bar\delta\geq0$, the relaxed region is
\[
\cR_{K,\bar\delta}(\mathcal I) = \left\{ G_K(z,m,a_K) \ \middle|\
    \begin{aligned}
      &\exists\,\xi^-,\xi^+\ \text{such that}\\
      &(z,m,a_K,\xi^-,\xi^+)
      \text{ is feasible for the relaxed system,}\\
      &\sum_{i\in\mathcal I}w_i(\xi_i^-+\xi_i^+)
      \leq\bar\delta
    \end{aligned}
\right\}.
\]
The region is an outer relaxation: it contains every pair generated by a distribution whose weighted quote loss is within the budget, while boundary points of the finite system need not be generated by such a distribution.  Because every nested quote set within a cross section is compared under the same $\bar\delta$, the resulting regions remain nested.  Appendix~\ref{app:empirical} describes the empirical choice of $\bar\delta$.

\subsubsection{Adaptive hull algorithm}

After constructing the applicable finite feasible set and checking Assumption~\ref{ass:interior} when required, the reconstruction uses the following support-function oracle.

\begin{description}[style=nextline,leftmargin=0pt,labelindent=0pt]
\item[Setup.] A feasible polytope $\cY$, the projection $G_K$, and the marginal box
\[
M_K=[\underline p,\overline p]\times[\underline L,\overline L],
\]
whose endpoints come from the four coordinate support bounds.  The initial direction set is
\[
D_0= \left\{ \left(\cos\frac{\pi j}{4},\sin\frac{\pi j}{4}\right): j=0,\ldots,7 \right\}.
\]
The numerical parameters are the degenerate-width tolerance $\varepsilon_{\rm deg}$, point tolerance $\varepsilon_{\rm pt}$, edge tolerance $\varepsilon_{\rm edge}$, area tolerances $(\varepsilon_{\rm area}^{\rm abs},\varepsilon_{\rm area}^{\rm rel})$, and the maximum number of support evaluations $N_{\max}$.  If either marginal range is at most $\varepsilon_{\rm deg}$, the image is classified as numerically degenerate and represented by a point or line segment.  Otherwise each coordinate is mapped to its marginal range,
\[
\widetilde p=\frac{p-\underline p}{\overline p-\underline p}, \qquad \widetilde L=\frac{L-\underline L}{\overline L-\underline L},
\]
so that the transformed marginal box is $[0,1]^2$.  In the remainder of this description, $x$, $V_t$, $H_t$, and $O_t$ are in these normalized coordinates.  A normalized direction $u=(u_p,u_L)$ corresponds to the original support direction
\[
a(u) = \left( \frac{u_p}{\overline p-\underline p}, \frac{u_L}{\overline L-\underline L} \right),
\]
and all inner products with the raw image $G_Ky$ use this direction.  The support program yields
\[
\bigl(x(u),\overline h(u)\bigr),
\]
where $x(u)$ is the normalized image of $G_Ky_u^\star$ and $y_u^\star$ is selected by a secondary program: for nonzero $u$ and $v=(-u_L,u_p)$, maximize $a(v)^\top G_Ky$ over $y\in\cY$ subject to $a(u)^\top G_Ky=\overline h_{\rm raw}(u)$, where $\overline h_{\rm raw}(u)$ is the value of the support program \eqref{eq:support} along $a(u)$ in the original coordinates.  The secondary program leaves the support value unchanged and, in exact arithmetic, returns an endpoint of the projected exposed face.  The primal--dual accuracy check uses the primary solve.  The normalized bound is
\[
\overline h(u) = \overline h_{\rm raw}(u) -\frac{u_p\underline p}{\overline p-\underline p} -\frac{u_L\underline L}{\overline L-\underline L}.
\]
In exact arithmetic, $\overline h(u)=\max u^\top x$ over the normalized image of $\cY$, and the primal--dual gap is zero.

\item[Initialization.] Evaluate the support program for each $u\in D_0$, set $V_0=\{x(u):u\in D_0\}$ after identifying points within $\varepsilon_{\rm pt}$, and let $d=\dim\operatorname{aff}(V_0)$.  If $d=0$, return the identified point.  If $d=1$, let $n\in\R^2$ be a unit normal to the candidate line and solve in directions $n$ and $-n$.  Add the two returned support points and recompute the affine dimension using $\varepsilon_{\rm pt}$.  If it remains one, return the two extreme endpoints; otherwise continue with $d=2$.

\item[Refinement.] At iteration $t$, form $H_t=\operatorname{conv}(V_t)$.  For every edge $e\in\mathcal E(H_t)$ with outward unit normal $n_e$ and endpoint $v_e$, solve in direction $n_e$ and compute the numerical edge gap
\[
\widehat\gamma_t(e) = \max\{0,\overline h(n_e)-n_e^\top v_e\}, \qquad \Gamma_t=\max_{e\in\mathcal E(H_t)}\widehat\gamma_t(e).
\]
Add every $x(n_e)$ satisfying $\widehat\gamma_t(e)>\varepsilon_{\rm edge}$ that is distinct within $\varepsilon_{\rm pt}$, and recompute the hull.  To preserve numerical nesting, replace each collected upper bound by
\[
\overline h_t^{\,\rm cons}(u) = \max\left\{\overline h(u),\max_{v\in H_t}u^\top v\right\}.
\]
Using all direction and upper-bound pairs collected through iteration $t$, form
\[
O_t = \operatorname{conv}\!\left( \left[ [0,1]^2\cap \bigcap_{\text{collected }(u,\overline h)} \{x:u^\top x\leq\overline h_t^{\,\rm cons}(u)\} \right] \cup H_t \right).
\]
These adjustments enlarge the numerical outer approximation and enforce nesting of the two computed polygons.  Containment of the true set additionally requires primal-feasible support points and valid dual upper bounds, including the bounds used for the initial marginal box.  Numerical nesting alone does not provide that certificate.

\item[Convergence.] The edge criterion is met when $\Gamma_t\leq\varepsilon_{\rm edge}$.  Under relaxation with the common budget, refinement may also end when
\[
\mathcal A_{p_0}(O_t)-\mathcal A_{p_0}(H_t) \leq \varepsilon_{\rm area}^{\rm abs} +\varepsilon_{\rm area}^{\rm rel}\mathcal A_{p_0}(O_t),
\]
where the area functional is evaluated after mapping $H_t$ and $O_t$ back to the original $(p,L)$ coordinates.  If the limit on support evaluations is reached first, or if a sweep with a positive gap adds no distinct point, the procedure ends without asserting the edge criterion.  Otherwise set $V_{t+1}$ to the expanded support set and repeat.

\item[Transformation.] For a nondegenerate return, map $H_t$ and $O_t$ back to the original $(p,L)$ coordinates as $H$ and $O$.  Map their edges by $\phi(p,L)=(p,L/p)$ and evaluate \eqref{eq:area}.
\end{description}

Primal and dual feasibility tolerances are $10^{-10}$.  The strict-interior cutoff is $10^{-8}$, and the degenerate-width, point, and normalized edge tolerances are $10^{-10}$, $10^{-9}$, and $10^{-8}$, respectively.  Relaxed cases may also stop when the absolute area gap is below $10^{-12}$ plus 0.2\% of outer area.  The edge criterion alone establishes the stated edge tolerance.

\subsubsection{Dual support program}
\label{app:dual}

Stack the mass and first-moment normalizations as $Ey=\mathbf q_E$ with $\mathbf q_E=(1,1)^\top$, and collect every remaining linear inequality, apart from nonnegativity, as $Ay\leq\mathbf q_A$.  Under the original intervals these are the cell-moment, threshold-atom, and bid--ask quote inequalities.  Under relaxation, the quote rows carry slacks and the common budget adds one row.  Let $G$ be the two-row matrix that implements $G_K$ on $y$, with zero columns for any slack variables.  For direction $u$, the dual of \eqref{eq:support} is
\[
\begin{aligned}
  \min_{\lambda,\nu}\quad
  &\mathbf q_A^\top\lambda+\mathbf q_E^\top\nu\\
  \text{subject to}\quad
  &A^\top\lambda+E^\top\nu\geq G^\top u,\\
  &\lambda\geq0,\qquad \nu\ \text{free}.
\end{aligned}
\]
Quote, cell, atom, and, when present, common budget constraints contribute distinct blocks of $\lambda$.

\subsubsection{Transformed area calculation}

For a polygonal set $\cC$, both section boundaries are affine between successive vertex probabilities.  If $L_{\cC}^\pm(p)=\alpha_\pm+\beta_\pm p$ on $[p_1,p_2]$, that interval contributes
\[
(\alpha_+-\alpha_-)\log\!\frac{p_2}{p_1} +(\beta_+-\beta_-)(p_2-p_1)
\]
to \eqref{eq:area}.  Summing these contributions evaluates the transformed area without numerical quadrature.

\subsubsection{Synthetic validation}

Five zero-width synthetic designs are evaluated at both $K=0.85$ and $K=0.90$.  For each $x\in\{K,K-10^{-6},K+10^{-6}\}$, one design assigns weights $(0.15,0.70,0.15)$ to $(x,1,2-x)$.  One design places a small weight far in the tail, assigning $(0.40,0.59,0.01)$ to $(0.9,1,5)$, and the bimodal design assigns equal weights to $(0.82,1.18)$.  Each distribution has mean one.  Across all ten combinations of a design and a threshold, the LP remains feasible when both known coordinates are imposed, and the known $(p_K,L_K)$ lies in the numerical outer envelope.  Point and line projections are classified as degenerate.  The zero-width designs check containment of known points in the finite system.  They do not verify equality with the original zero-width identified set, since Assumption~\ref{ass:interior} does not apply.  The positive-width empirical cases use that assumption.

\subsection{Supplementary empirical evidence}
\label{app:empirical}

\subsubsection{Sampling details}

Within each calendar week, Wednesday is preferred.  If it is absent or ineligible, the fallback order is Tuesday, Monday, Thursday, Friday.  Ties between eligible expiries are resolved in favor of the earlier expiry.  For the put pair, if both selections coincide with a quote exactly at $K$, the pair retains that strike and the next-nearest distinct strike by absolute distance.

\subsubsection{Qualification and numerical verification}

Before contract-identity screening, 2,076 feasible cases satisfy the strict interior cutoff, with a smallest relative margin of 0.02.  The remaining 16 feasible cases have zero maximum margin, below the numerical cutoff of $10^{-8}$.

The common relaxation budget is the larger of the 95th percentile of $\delta^\star$ in the calibration sample (zero) and the prespecified floor based on normalized quote tick size.  This gives $\bar\delta=5\times10^{-5}$, fixed before analysis of the main sample.  Every nondegenerate case under the original constraints satisfies the normalized edge criterion.  Replacing midpoint areas by the endpoints of their numerical envelopes leaves the reported medians unchanged at displayed precision.

\subsubsection{Subsample stability and robustness}

\begin{table}[H]
\centering
\caption{Primary-sample stability of area ratios under the complete put wing}
\label{tab:subsample}
\begin{tabular}{lrr}
\toprule
Subsample & Nondegenerate cases & Median area ratio\\
\midrule
All primary cases & 1,966 & 0.6340\\
Excluding calibration dates & 1,946 & 0.6338\\
2013--2015 & 578 & 0.6233\\
2016--2018 & 542 & 0.6336\\
2019--2021 & 546 & 0.6432\\
2022--August 2023 & 300 & 0.6367\\
\bottomrule
\end{tabular}
\begin{minipage}{0.90\textwidth}
\footnotesize\emph{Notes:} All rows use the primary sample definition.  Case counts report observations entering the area-ratio median.  Calibration dates are those used to set the common budget before analysis of the main sample.  Lower ratios indicate tighter joint sets relative to the marginal benchmark, and similar values across rows indicate stability over time and sample composition.
\end{minipage}
\end{table}

\begin{table}[H]
\centering
\caption{Range of median area ratios across horizon and threshold cells}
\label{tab:robustness}
\begin{tabular}{lrr}
\toprule
Specification & Minimum median & Maximum median\\
\midrule
Forward $-20$ bp & 0.6200 & 0.6436\\
Forward $+20$ bp & 0.6225 & 0.6458\\
Doubled bid--ask spread & 0.6053 & 0.6458\\
50\% strike thinning & 0.6162 & 0.6378\\
Narrowest quoted interval per strike & 0.6237 & 0.6462\\
$p_0=10^{-5}$ & 0.6225 & 0.6441\\
$p_0=10^{-3}$ & 0.6373 & 0.6639\\
Long maturity & 0.6336 & 0.6527\\
Threshold $K=0.80$ & 0.6128 & 0.6316\\
Threshold $K=0.95$ & 0.6396 & 0.6556\\
Full call--put cross section & 0.6301 & 0.6926\\
\bottomrule
\end{tabular}
\begin{minipage}{0.90\textwidth}
\footnotesize\emph{Notes:} These sensitivity summaries pool feasible and relaxed cases after contract screening.  Entries based on the original intervals have a sharp interpretation when Assumption~\ref{ass:interior} holds; common-budget entries describe the outer relaxations in Appendix~\ref{app:relax}.  The pooled summaries are sensitivity diagnostics rather than estimates for the primary sharp sample alone.  Strike thinning retains every other distinct strike after sorting from lowest to highest.  The long-maturity check uses an eligible expiry in 61--120 days closest to 90 days.  Each ratio compares a specification with its own marginal benchmark, so the rows measure stability rather than relative performance.
\end{minipage}
\end{table}

\subsubsection{Price-grid benchmark}

For support cap $x_{\max}\in\{2,3\}$ and $N_{\rm grid}$ nodes, the base grid is $x_j=jx_{\max}/(N_{\rm grid}-1)$ for $j=0,\ldots,N_{\rm grid}-1$.  The inserted-node design replaces the nearest interior node by $K$ and reorders the grid, resolving an equal-distance tie in favor of the smaller node.  The centered-cell design uses interior nodes $K+(\ell+\tfrac12)x_{\max}/(N_{\rm grid}-1)$ for $\ell\in\mathbb Z$, retaining the $N_{\rm grid}-2$ nearest points in $(0,x_{\max})$ and adding $0$ and $x_{\max}$.

Let $c$ index a case and $g$ a grid specification.  Let $\delta^\star_{{\rm grid},g,c}$ be the minimum weighted quote loss for the grid formulation, and let $\delta^\star_{{\rm cont},c}$ be the continuous-state minimum computed before imposing the support cap.  The shared budget is
\begin{equation}
  \delta_{g,c}
  =
  \max\{\delta^\star_{{\rm grid},g,c},
        \delta^\star_{{\rm cont},c}\}
  +10^{-10}.
  \label{eq:gridbudget}
\end{equation}
Once the shared budget in \eqref{eq:gridbudget} is chosen from the grid and uncapped continuous minima, the matched continuous support program imposes $0\leq X\leq x_{\max}$, equivalently $m_J\leq x_{\max}z_J$ in the terminal cell.

Directions are
\[
D_{64}= \left\{ \left(\cos\frac{\pi\ell}{32},\sin\frac{\pi\ell}{32}\right): \ell=0,\ldots,63 \right\}.
\]
For case $c$, grid $g$, and $u\in D_{64}$, let $h_{g,c}(u)$ and $h_{\mathrm{cont},c}(u)$ be the support values from the grid and matched continuous-state programs, both under the shared budget $\delta_{g,c}$.  The reported absolute error is
\[
\begin{aligned}
  e_{g,c}(u)
  &=|h_{g,c}(u)-h_{\mathrm{cont},c}(u)|,\\
  e^{\max}_{g,c}
  &=\max_{u\in D_{64}}e_{g,c}(u).
\end{aligned}
\]
For each case, the median over directions is $\operatorname{med}_{u\in D_{64}} e_{g,c}(u)$; figures and tables report its median across cases, and the tails of the case-level distribution summarize $e^{\max}_{g,c}$.  The four benchmark synthetic cases at $K=0.90$ use the three threshold locations and the bimodal law defined in the synthetic validation subsection.  For these zero-width inputs, the continuous-state program under \eqref{eq:gridbudget} is the reference finite system.  The comparison does not invoke Theorem~\ref{thm:sharp} for the original zero-width quote problem.

\begin{table}[H]
\centering
\caption{Selected price-grid errors and computation times}
\label{tab:grid}
\begin{tabular}{rrlrrrr}
\toprule
Nodes & Cap & Threshold placement & Median error & Maximum error &
Grid budget & Seconds\\
\midrule
51  & 2 & Centered cell & 0.01136 & 0.10134 & $1.63\times10^{-4}$ & 0.138\\
51  & 2 & Inserted node & 0.00956 & 0.15000 & $1.22\times10^{-4}$ & 0.152\\
201 & 2 & Centered cell & 0.00122 & 0.07739 & $1.0\times10^{-10}$ & 0.419\\
201 & 2 & Inserted node & $1.52\times10^{-10}$ & 0.15000 & $1.0\times10^{-10}$ & 0.382\\
401 & 2 & Centered cell & 0.00026 & 0.07597 & $1.0\times10^{-10}$ & 0.781\\
401 & 2 & Inserted node & $<10^{-14}$ & 0.14999 & $1.0\times10^{-10}$ & 0.808\\
\bottomrule
\end{tabular}
\begin{minipage}{0.94\textwidth}
\footnotesize\emph{Notes:} Grid budget is the median additional weighted quote loss beyond the uncapped continuous-state minimum, so the reported error includes this feasibility adjustment, especially at 51 nodes.  Lower errors indicate closer agreement with the continuous-state solution, and a lower grid budget indicates less additional quote relaxation.  Maximum error is the largest $e^{\max}_{g,c}$ across cases.  Seconds report median computation time for each grid specification.
\end{minipage}
\end{table}

\bibliographystyle{plainnat}
\bibliography{references}

\end{document}